\documentclass[a4paper,USenglish,cleveref,thm-restate]{lipics-v2021}
\usepackage{style}
\usepackage{shortcuts}

\nolinenumbers
\hideLIPIcs

\usepackage{etoolbox}
\makeatletter
\patchcmd{\endabstract}{%
    \vskip\topmattervskip\baselineskip\noindent
    \subjclassHeading
    \ifx\@ccsdescString\@empty
        \textcolor{red}{Author: Please fill in 1 or more \string\ccsdesc\space macro}%
    \else
        \@ccsdescString
    \fi
    \vskip\topmattervskip\baselineskip
    \noindent\keywordsHeading
    \ifx\@keywords\@empty
        \textcolor{red}{Author: Please fill in \string\keywords\space macro}%
    \else
        \@keywords
    \fi}{}{}{}
\makeatother

\title{An Optimal Algorithm for Binary Closest String}
\author{Nick Fischer}{Max Planck Institute for Informatics, Saarbrücken, Germany}{nfischer@mpi-inf.mpg.de}{https://orcid.org/0009-0001-0909-3296}{}
\author{Mursalin Habib}{Rutgers University, New Brunswick, NJ, USA}{mursalin.habib@rutgers.edu}{https://orcid.org/0000-0002-6671-4669}{Supported by the National Science Foundation under Grants CCF-2313372, CCF-2422558, and CCF-2443697.}
\authorrunning{N. Fischer and M. Habib}
\funding{Part of this work was done while visiting INSAIT at Sofia University ``St. Kliment Ohridski'' and funded by the Bulgarian Ministry of Education and Science through its support of INSAIT under the Bulgarian National Roadmap for Research Infrastructure.}

\acknowledgements{We are very grateful to Karthik C.S.\ for fruitful discussions at an early stage of this project.}


\begin{document}

\maketitle

\begin{abstract}
We revisit the \emph{Binary Closest String} problem, which asks, given a set of binary strings $X \subseteq \{0, 1\}^n$, to compute a string minimizing the maximum Hamming distance to $X$. A long line of work has focused on parameterized algorithms with respect to the optimal distance $d$, yielding a sequence of improvements from $O^*(d^d)$ through $O^*(16^d)$, $O^*(9.513^d)$, $O^*(8^d)$, $O^*(6.731^d)$ to the current best-known running time of $O^*(5^d)$ [Chen, Ma, Wang; Algorithmica '16].

We present a faster randomized algorithm running in time $O^*(4^d)$. Our result matches a recent fine-grained lower bound [Abboud, Fischer, Goldenberg, Karthik C.S., Safier; ESA ’23], and is therefore conditionally \emph{optimal}. As an extra benefit, our algorithm is remarkably simple.
\end{abstract}

\section{Introduction}
The \emph{Closest String} problem asks, given a set of strings $X \subseteq \Sigma^n$, to find a string $y^* \in \Sigma^n$ that minimizes the maximum Hamming distance to the strings in $X$. It constitutes one of the most fundamental clustering tasks on strings, and is also of great importance in bioinformatics, for instance in the computer-aided design of universal PCR primers~\cite{LanctotLMWZ03}, motif finding~\cite{GrammHN02}, and drug design~\cite{DengLLMW03}.

The Closest String problem is extremely well-studied, from both the theory and applied computational biology communities. It is known to be NP-hard~\cite{FrancesL97,LanctotLMWZ03} (even for binary strings), and well-understood through the lens of approximation~\cite{Ben-DorLPR97,GasieniecJL99,LanctotLMWZ03,LiMW02,AndoniIP06,MaS09}. Here we focus on the parameterized setting, where the goal is to design algorithms whose running time depends on the optimal distance $d$. This setting has concrete practical relevance, as for the aforementioned biology applications $d$ is typically very small. Over the past two decades, a long sequence of papers has developed progressively faster parameterized algorithms~\cite{GrammNR03,MaS09,ChenW11,WangZ09,ZhaoZ10,ChenMW12,ChenMW16}. And yet, despite this sustained effort, it remained open to pinpoint the precise complexity of Closest String---even in the most basic regime of binary strings (i.e., for $|\Sigma| = 2$).

Concretely, the first parameterized algorithm was given by Gramm, Niedermeier, and Rossmanith~\cite{GrammNR03} and runs in time\footnote{Here and throughout, the $O^*(\cdot)$ notation hides polynomial factors in the input size, $\poly(|X|, n)$.} $O^*(d^d)$. The subsequent milestones for Binary Closest String are algorithms in time $O^*(16^d)$~\cite{MaS09}, $O^*(9.513^d)$~\cite{WangZ09}, $O^*(8^d)$~\cite{ZhaoZ10}, $O^*(6.731^d)$~\cite{ChenMW12}, before culminating in the current state-of-the-art $O^*(5^d)$ due to Chen, Ma, and Wang~\cite{ChenMW16}. On the converse side, a recent lower bound construction due to Abboud, Fischer, Goldenberg, Karthik~C.S., and Safier~\cite{AbboudFGSS23} implies that there is no $O^*((4-\epsilon)^d)$-time algorithm for Binary Closest String, for any constant~\makebox{$\epsilon > 0$}, unless the Strong Exponential-Time Hypothesis (SETH) from fine-grained complexity fails.\footnote{In detail: Abboud, Fischer, Goldenberg, Karthik~C.S., and Safier~\cite{AbboudFGSS23} show that Binary Closest String cannot be solved in time $O^*((2-\epsilon)^n)$. Inspecting their lower bound reveals that it applies to instances with $d = (1/2 \pm o(1)) n$, implying that Binary Closest String cannot be solved in time \smash{$O^*((4-\epsilon)^d)$} in the regime where $d = (1/2 \pm o(1)) n$. A simple padding argument extends this lower bound to the full range $d \leq n/2$. For $d > n/2$ the brute-force $O^*(2^n)$-time algorithm becomes more efficient.}

This leaves an exponential gap between the $5^d$ upper bound and the $4^d$ lower bound. In this paper we finally close this gap, and thereby resolve the parameterized complexity of Binary Closest String:

\begin{theorem} \label{thm:main}
Binary Closest String can be solved in randomized time $O^*(4^d)$.
\end{theorem}

Notably, our algorithm is not only conditionally optimal, but also pleasingly \emph{simple}. It can be stated in 6 lines of pseudocode (see \cref{alg:main}), and is also conceptually simple. In fact, our algorithm is almost identical to a simple local search algorithm by Chen, Ma, and Wang~\cite{ChenMW16} that runs in time $O^*((2e)^d) = O^*(5.437^d)$  (from same paper that established the previous state of the art). We propose a technical yet important nuance: instead of running~$O^*((2e)^d)$ independent local searches of length $O(d)$, we run \emph{one} local search of length~$O^*(4^d)$. Our analysis is based on Markov chains and conceptually inspired by Schöning's seminal algorithm~\cite{Schoning99} for $k$-SAT. The specific Markov chain turns out to be more complicated than in Schöning's algorithm, but viewed the right way its analysis does not take more than a page.

To conclude, while our new algorithm will perhaps not have immediate practical impact, it does constitute the first optimal result in a long line of work~\cite{GrammNR03,MaS09,ChenW11,WangZ09,ZhaoZ10,ChenMW12,ChenMW16}. For this reason we are confident that it may form the baseline to make progress on more general variants, e.g., for larger alphabets or for the strictly more general and even more practically important \emph{Closest Substring} problem.

\section{Formal Definitions}
The \emph{Hamming distance} of two strings $x, y \in \Sigma^n$ is defined as \makebox{$\HD(x, y) = |\set{ j : x[j] \neq y[j] }|$}. The \emph{Closest String} problem is, given a set $X \subseteq \Sigma^n$, to find a string $y^*$ that minimizes $\max_{x \in X} \HD(x, y^*)$. We define the parameter $d = \max_{x \in X} \HD(x, y^*)$ as the optimal distance. The \emph{Binary Closest String} problem is the restriction to alphabets of size $|\Sigma| = 2$.

\section{Algorithm}
This section is devoted to the proof of \cref{thm:main}. Throughout, to solve the Closest String problem, we assume that as part of the input we are given a parameter $1 \leq d < n$ and the goal is to decide if there is a center string $y^* \in \set{0, 1}^n$ within distance $\max_{x \in X} \HD(x, y) \leq d$. The optimal $d$ can then be identified by testing $d \gets 1, 2, \dots$ one by one; the overhead of this search is only a constant factor in our final $O^*(4^d)$-time algorithm.

\begin{algorithm}[t]
\caption{The $O^*(4^d)$-time algorithm for Binary Closest String} \label{alg:main}
\begin{algorithmic}[1]
    \Input{A set of strings $X \subseteq \set{0, 1}^n$ and an integer $d \geq 0$}
    \Output{A string $y^* \in \set{0, 1}^n$ with $\max_{x \in X} \HD(x, y^*) \leq d$ (if it exists)}
    \medskip

    \State Let $y \in X$ be arbitrary
    \RepeatInf
        \State Choose any string $x \in X$ maximizing $\HD(x, y)$
        \If{$\HD(x, y) \leq d$} \Return $y$\EndIf
        \State Select an index $j \in [n]$ with $x[j] \neq y[j]$ uniformly at random
        \State Set $y[j] \gets x[j]$
    \EndRepeatInf
\end{algorithmic}
\end{algorithm}

Consider the pseudocode in \cref{alg:main}. The algorithm is designed in such a way that it will never return an incorrect output. In particular, if there is no solution then the algorithm does not terminate. In the following we assume that there is a solution $y^*$, and argue that the algorithm will terminate in time $O^*(4^d)$ with high probability.

Throughout fix an optimal solution $y^*$ with $\max_{x \in X} \HD(x, y^*) \leq d$ (if there are several, fix an arbitrary one). We say that the algorithm is in \emph{state} $i$ if~\makebox{$\HD(y, y^*) = i$}. Initially, the algorithm starts in some state $i \leq d$ (as $\HD(y, y^*) \leq d$ for whichever element~\makebox{$y \in X$} we picked). In each iteration the algorithm flips exactly one bit in $y$, and thus the current state~$i$ either increases or decreases by exactly $1$. Once the algorithm reaches state $0$ it certainly terminates (though it might also terminate before if there is more than one optimal solution). The following lemma is based on the ideas from~\cite{ChenMW16}.

\begin{lemma} \label{lem:progress}
In any iteration with state $i > 0$ the algorithm transitions to state $i-1$ with probability at least $p_i = \min\set{\frac{i}{2d}, \frac{1}{2}}$.
\end{lemma}
\begin{proof}
Let $y, x$ be as in the current iteration. We distinguish three types of indices $j$: The \emph{good} ones with $x[j] = y^*[j] \neq y[j]$, the \emph{bad} ones with $x[j] \neq y[j] = y^*[j]$, and the \emph{useless} ones with $x[j] = y[j] \neq y^*[j]$. (As the alphabet is binary, this only leaves the indices $j$ with $x[j] = y[j] = y^*[j]$ which play no role in the analysis.) Let $G, B, U$ denote the number of good, bad and useless indices, respectively. Observe that
\begin{alignat*}{3}
    &\HD(x, y) &&= G + B &&> d, \\
    &\HD(x, y^*) &&= B + U &&\leq d, \\
    &\HD(y, y^*) &&= G + U &&= i.
\end{alignat*}
By design the algorithm picks a uniformly random index among the good and bad ones, and succeeds in decreasing the state of the algorithm if it picks a good one. In other words, the probability to move to state $i - 1$ is exactly
\begin{equation*}
    \frac{G}{G + B} = \frac{(G + B) - (B + U) + (G + U)}{2(G + B)} \geq \frac{(G + B) - d + i}{2(G + B)}.
\end{equation*}
If $d \geq i$, then the function $\frac{x - d+i}{2x}$ is non-decreasing for positive $x$. If $d < i$, the function is instead decreasing, but always lower-bounded by $\frac12$. It follows that
\begin{equation*}
    \frac{G}{G + B} \geq \min\set*{\frac{d - d + i}{2d}, \frac12} = \min\set*{\frac{i}{2d}, \frac12},
\end{equation*}
which is as claimed.
\end{proof}

\cref{lem:progress} allows us to model the execution by means of a \emph{Markov chain} with transition probabilities as described in the lemma. More precisely, the textbook approach is to define a Markov chain that is \emph{coupled} with the state of the algorithm, and to analyze the hitting time of our target state $0$ in that coupled Markov chain. However, to minimize the notational overhead we here prefer a direct approach.

\begin{lemma} \label{lem:hitting-time}
The expected number of iterations to reach state $0$ is at most $4^d \cdot n^2$.
\end{lemma}
\begin{proof}
Let $D_i$ denote the expected number of decrements (i.e., iterations that decrease the state) to move from state~$i$ to state~\makebox{$i-1$}. We trivially have $D_n = 1$. For any state $0 < i < n$, \cref{lem:progress} shows that with probability at least~$p_i$ we move to $i-1$ in one transition. Otherwise, with probability at most $1 - p_i$, we first transition to $i+1$ and then the expected number of decrements to visit state $i - 1$ is at most $D_{i+1} + D_i$. This yields the recurrence
\begin{equation*}
    D_i \leq p_i \cdot 1 + (1 - p_i) \cdot (D_{i+1} + D_i),
\end{equation*}
or equivalently,
\begin{equation*}
    D_i \leq 1 + \frac{1 - p_i}{p_i} \cdot D_{i+1}.
\end{equation*}
Unrolling the recurrence gives, for all $0 < i \leq n$,
\begin{equation*}
    D_i \leq \sum_{j=i}^n \prod_{m=i}^{j-1} \frac{1-p_m}{p_m}.
\end{equation*}
Now recall that $p_m = \min\set{\frac{m}{2d}, \frac{1}{2}}$. Thus, for $m \geq d$ the quotient $\frac{1-p_m}{p_m}$ becomes exactly $1$, and for $m < d$ it is certainly at least $1$. It follows that
\begin{equation*}
    D_i \leq \sum_{j=i}^n \prod_{m=i}^{d} \frac{1-p_m}{p_m} \leq n \prod_{m=1}^d \frac{1-p_m}{p_m} = n \prod_{m=1}^d \frac{2d - m}{m} = n \cdot \binom{2d - 1}{d}.
\end{equation*}
Therefore, the expected number of decrements to reach state $0$ from an arbitrary starting state~$i_0$ is at most
\begin{equation*}
    \sum_{i=1}^{i_0} D_i \leq \sum_{i=1}^n D_i \leq n^2 \cdot \binom{2d-1}{d} \leq \frac{4^d \cdot n^2}{2}.
\end{equation*}
Finally, in any sequence from state $i_0$ to $0$ the total number of increments is always upper bounded by the number of decrements, and so the expected total number of iterations is at most $4^d \cdot n^2$.
\end{proof}

The proof of \cref{thm:main} is now a simple consequence of \cref{lem:hitting-time}: Within $4^d \cdot n^2$ iterations the algorithm is expected to reach state $0$. Thus, within $4^d \cdot 2n^2$ iterations the algorithm terminates with probability at least $1/2$ (by Markov's inequality), and within $4^d \cdot n^3$ iterations the algorithm fails to terminate with exponentially small error probability $\exp(-\Omega(n))$. In particular, after $4^d \cdot n^3$ iterations we can interrupt \cref{alg:main} to confidently claim that there is no center in distance~$d$. Each iteration runs in polynomial time, and thus the total running time is $O^*(4^d)$.

\begin{remark}
We finally comment on a micro-optimization of \cref{alg:main}. The term $n^2$ in \cref{lem:hitting-time} (and in the resulting running time) can easily be replaced by $d^2$ by making sure that the algorithm only traverses the states $\set{0, \dots, O(d)}$. To achieve this, one option is to add another line to the algorithm: When $\max_{x \in X} \HD(x, y) > 2d$ then reset $y$ to an arbitrary string in $X$. Indeed, in this case the triangle inequality implies $\HD(y, y^*) > d$, and hence the reset never increases the state. Combined with a careful implementation of \cref{alg:main}, this results in an algorithm with expected running time $O(n |X| + 4^d d^2 |X|)$.
\end{remark}

\bibliographystyle{plainurl}
\bibliography{main}

@inproceedings{Schoning99,
  author       = {Uwe Sch{\"{o}}ning},
  title        = {A Probabilistic Algorithm for k-SAT and Constraint Satisfaction Problems},
  booktitle    = {40th {IEEE} Symposium on Foundations of Computer Science ({FOCS} 1999)},
  pages        = {410--414},
  publisher    = {{IEEE} Computer Society},
  year         = {1999},
  url          = {https://doi.org/10.1109/SFFCS.1999.814612},
  doi          = {10.1109/SFFCS.1999.814612},
  bibsource    = {dblp computer science bibliography, https://dblp.org}
}

@inproceedings{AndoniIP06,
  author       = {Alexandr Andoni and
                  Piotr Indyk and
                  Mihai P{\u{a}}tra{\c{s}}cu},
  title        = {On the Optimality of the Dimensionality Reduction Method},
  booktitle    = {47th {IEEE} Symposium on Foundations of Computer Science ({FOCS} 2006)},
  pages        = {449--458},
  publisher    = {{IEEE} Computer Society},
  year         = {2006},
  url          = {https://doi.org/10.1109/FOCS.2006.56},
  doi          = {10.1109/FOCS.2006.56},
  bibsource    = {dblp computer science bibliography, https://dblp.org}
}

@inproceedings{GasieniecJL99,
  author       = {Leszek Gasieniec and
                  Jesper Jansson and
                  Andrzej Lingas},
  editor       = {Robert Endre Tarjan and
                  Tandy J. Warnow},
  title        = {Efficient Approximation Algorithms for the Hamming Center Problem},
  booktitle    = {10th {ACM-SIAM} Symposium on Discrete Algorithms ({SODA} 1999)},
  pages        = {905--906},
  publisher    = {{ACM/SIAM}},
  year         = {1999},
  url          = {http://dl.acm.org/citation.cfm?id=314500.315081},
  bibsource    = {dblp computer science bibliography, https://dblp.org}
}

@inproceedings{AbboudFGSS23,
  author       = {Amir Abboud and
                  Nick Fischer and
                  Elazar Goldenberg and
                  {Karthik {C. S.}} and
                  Ron Safier},
  editor       = {Inge Li G{\o}rtz and
                  Martin Farach{-}Colton and
                  Simon J. Puglisi and
                  Grzegorz Herman},
  title        = {Can You Solve Closest String Faster Than Exhaustive Search?},
  booktitle    = {31st Annual European Symposium on Algorithms, {ESA} 2023, Amsterdam,
                  The Netherlands, September 4-6, 2023},
  series       = {LIPIcs},
  volume       = {274},
  pages        = {3:1--3:17},
  publisher    = {Schloss Dagstuhl - Leibniz-Zentrum f{\"{u}}r Informatik},
  year         = {2023},
  url          = {https://doi.org/10.4230/LIPIcs.ESA.2023.3},
  doi          = {10.4230/LIPICS.ESA.2023.3},
  bibsource    = {dblp computer science bibliography, https://dblp.org}
}

@inproceedings{Ben-DorLPR97,
  author       = {Amir Ben{-}Dor and
                  Giuseppe Lancia and
                  Jennifer Perone and
                  R. Ravi},
  editor       = {Alberto Apostolico and
                  Jotun Hein},
  title        = {Banishing Bias from Consensus Sequences},
  booktitle    = {8th Annual Symposium on Combinatorial Pattern Matching ({CPM} 1997)},
  series       = {Lecture Notes in Computer Science},
  volume       = {1264},
  pages        = {247--261},
  publisher    = {Springer},
  year         = {1997},
  url          = {https://doi.org/10.1007/3-540-63220-4\_63},
  doi          = {10.1007/3-540-63220-4\_63},
  bibsource    = {dblp computer science bibliography, https://dblp.org}
}

@article{LiMW02,
  author       = {Ming Li and
                  Bin Ma and
                  Lusheng Wang},
  title        = {On the closest string and substring problems},
  journal      = {J. {ACM}},
  volume       = {49},
  number       = {2},
  pages        = {157--171},
  year         = {2002},
  url          = {https://doi.org/10.1145/506147.506150},
  doi          = {10.1145/506147.506150},
  bibsource    = {dblp computer science bibliography, https://dblp.org}
}

@article{DengLLMW03,
  author       = {Xiaotie Deng and
                  Guojun Li and
                  Zimao Li and
                  Bin Ma and
                  Lusheng Wang},
  title        = {Genetic Design of Drugs Without Side-Effects},
  journal      = {{SIAM} J. Comput.},
  volume       = {32},
  number       = {4},
  pages        = {1073--1090},
  year         = {2003},
  url          = {https://doi.org/10.1137/S0097539701397825},
  doi          = {10.1137/S0097539701397825},
  bibsource    = {dblp computer science bibliography, https://dblp.org}
}

@article{MaS09,
  author       = {Bin Ma and
                  Xiaoming Sun},
  title        = {More Efficient Algorithms for Closest String and Substring Problems},
  journal      = {{SIAM} J. Comput.},
  volume       = {39},
  number       = {4},
  pages        = {1432--1443},
  year         = {2009},
  url          = {https://doi.org/10.1137/080739069},
  doi          = {10.1137/080739069},
  bibsource    = {dblp computer science bibliography, https://dblp.org}
}

@article{ChenMW16,
  author       = {Zhi{-}Zhong Chen and
                  Bin Ma and
                  Lusheng Wang},
  title        = {Randomized Fixed-Parameter Algorithms for the Closest String Problem},
  journal      = {Algorithmica},
  volume       = {74},
  number       = {1},
  pages        = {466--484},
  year         = {2016},
  url          = {https://doi.org/10.1007/s00453-014-9952-y},
  doi          = {10.1007/S00453-014-9952-Y},
  bibsource    = {dblp computer science bibliography, https://dblp.org}
}

@article{ChenMW12,
  author       = {Zhi{-}Zhong Chen and
                  Bin Ma and
                  Lusheng Wang},
  title        = {A three-string approach to the closest string problem},
  journal      = {J. Comput. Syst. Sci.},
  volume       = {78},
  number       = {1},
  pages        = {164--178},
  year         = {2012},
  url          = {https://doi.org/10.1016/j.jcss.2011.01.003},
  doi          = {10.1016/J.JCSS.2011.01.003},
  bibsource    = {dblp computer science bibliography, https://dblp.org}
}

@article{GrammNR03,
  author       = {Jens Gramm and
                  Rolf Niedermeier and
                  Peter Rossmanith},
  title        = {Fixed-Parameter Algorithms for {CLOSEST} {STRING} and Related Problems},
  journal      = {Algorithmica},
  volume       = {37},
  number       = {1},
  pages        = {25--42},
  year         = {2003},
  url          = {https://doi.org/10.1007/s00453-003-1028-3},
  doi          = {10.1007/S00453-003-1028-3},
  bibsource    = {dblp computer science bibliography, https://dblp.org}
}

@article{FrancesL97,
  author       = {Moti Frances and
                  Ami Litman},
  title        = {On Covering Problems of Codes},
  journal      = {Theory Comput. Syst.},
  volume       = {30},
  number       = {2},
  pages        = {113--119},
  year         = {1997},
  url          = {https://doi.org/10.1007/s002240000044},
  doi          = {10.1007/S002240000044},
  bibsource    = {dblp computer science bibliography, https://dblp.org}
}

@article{LanctotLMWZ03,
  author       = {J. Kevin Lanct{\^{o}}t and
                  Ming Li and
                  Bin Ma and
                  Shaojiu Wang and
                  Louxin Zhang},
  title        = {Distinguishing string selection problems},
  journal      = {Inf. Comput.},
  volume       = {185},
  number       = {1},
  pages        = {41--55},
  year         = {2003},
  url          = {https://doi.org/10.1016/S0890-5401(03)00057-9},
  doi          = {10.1016/S0890-5401(03)00057-9},
  bibsource    = {dblp computer science bibliography, https://dblp.org}
}

@inproceedings{WangZ09,
  author       = {Lusheng Wang and
                  Binhai Zhu},
  title        = {Efficient Algorithms for the Closest String and Distinguishing String Selection Problems},
  booktitle    = {Frontiers in Algorithmics},
  year         = {2009},
  publisher    = {Springer Berlin Heidelberg},
  pages        = {261--270},
  doi          = {10.1007/978-3-642-02270-8_27},
}

@inproceedings{ZhaoZ10,
  author       = {Ruixuan Zhao and
                  Ning Zhang},
  title        = {A More Efficient Closest String Problem.},
  booktitle    = {2nd International Conference on Bioinformatics and Computational Biology},
  year         = {2010},
  pages        = {210--215},
}

@article{ChenW11,
  author       = {Zhi{-}Zhong Chen and
                  Lusheng Wang},
  title        = {Fast Exact Algorithms for the Closest String and Substring Problems
                  with Application to the Planted (L, d)-Motif Model},
  journal      = {{IEEE} {ACM} Trans. Comput. Biol. Bioinform.},
  volume       = {8},
  number       = {5},
  pages        = {1400--1410},
  year         = {2011},
  url          = {https://doi.org/10.1109/TCBB.2011.21},
  doi          = {10.1109/TCBB.2011.21},
  bibsource    = {dblp computer science bibliography, https://dblp.org}
}

@inproceedings{GrammHN02,
  title        = {Closest strings, primer design, and motif search},
  author       = {Jens Gramm and
                  Falk H{\"u}ffner and
                  Rolf Niedermeier},
  booktitle    = {Currents in Computational Molecular Biology, poster abstracts of RECOMB},
  volume       = {2002},
  pages        = {74--75},
  year         = {2002},
}

\end{document}